\documentclass[11.5pt]{article}
\usepackage{url,ifthen}
\usepackage{srcltx}
\usepackage{multirow}
\usepackage{boxedminipage}
\usepackage[margin=1.1in]{geometry}
\usepackage{nicefrac}
\usepackage{xspace}
\usepackage{graphicx}
\usepackage{srcltx}
\usepackage[usenames]{color}
\usepackage{fullpage}
\usepackage{xspace}
\definecolor{DarkGreen}{rgb}{0.1,0.5,0.1}
\definecolor{DarkRed}{rgb}{0.5,0.1,0.1}
\definecolor{DarkBlue}{rgb}{0.1,0.1,0.5}
\usepackage[small]{caption}
\usepackage{float}
\usepackage{balance}
\usepackage{amsmath}
\usepackage{amsfonts}
\usepackage{amssymb}
\usepackage{amsthm}

\newtheorem{theorem}{Theorem}[section]
\newtheorem*{namedtheorem}{\theoremname}
\newcommand{\theoremname}{testing}

\newtheorem{lemma}[theorem]{Lemma}

\newtheorem{claim}[theorem]{Claim}

\newtheorem{fact}[theorem]{Fact}

\newtheorem{observation}{Observation}
\newtheorem{corollary}[theorem]{Corollary}

\newtheorem*{question*}{Question}

\theoremstyle{definition}
\newtheorem{definition}[theorem]{Definition}

\theoremstyle{plain}
\newtheorem{Alg}{Algorithm}

\definecolor{DarkGreen}{rgb}{0.1,0.5,0.1}
\definecolor{DarkRed}{rgb}{0.5,0.1,0.1}
\definecolor{DarkBlue}{rgb}{0.1,0.1,0.5}

\usepackage[pdftex]{hyperref}
\hypersetup{
    unicode=false,          
    pdftoolbar=true,        
    pdfmenubar=true,        
    pdffitwindow=false,      
    pdfnewwindow=true,      
    colorlinks=true,       
    linkcolor=DarkGreen,          
    citecolor=DarkGreen,        
    filecolor=DarkGreen,      
    urlcolor=DarkBlue,          
}

\newcommand{\ignore}[1]{}

\newcommand{\R}{\mathbb R}

\makeatletter
\floatstyle{ruled}
\newfloat{fragment}{H}{lop}
\floatname{fragment}{Algorithm}
\renewcommand{\floatc@ruled}[2]{\vspace{2pt}{\@fs@cfont \#1.\:} \#2 \par
 \vspace{1pt}}
\makeatother

\title{Super-resolution, Extremal Functions and the Condition Number of Vandermonde Matrices}
\author{Ankur Moitra \thanks{
Massachusetts Institute of Technology. Department of Mathematics and the Computer Science and Artificial Intelligence Lab. Email: {\tt moitra@mit.edu}.
This work is supported in part by a grant from the MIT NEC Corporation and a Google Research Award. }}

\begin{document}
\maketitle

\begin{abstract}
Super-resolution is a fundamental task in imaging, where the goal is to extract fine-grained structure from coarse-grained measurements. Here we are interested in a popular mathematical abstraction of this problem that has been widely studied in the statistics, signal processing and machine learning communities. We exactly resolve the threshold at which noisy super-resolution is possible. In particular, we establish a sharp phase transition for the relationship between the cutoff frequency ($m$) and the separation ($\Delta$). If $m > 1/\Delta + 1$, our estimator converges to the true values at an inverse polynomial rate in terms of the magnitude of the noise. And when $m < (1-\epsilon) /\Delta$ no estimator can distinguish between a particular pair of $\Delta$-separated signals even if the magnitude of the noise is exponentially small. 

Our results involve making novel connections between {\em extremal functions} and the spectral properties of Vandermonde matrices. We establish a sharp phase transition for their condition number which in turn allows us to give the first noise tolerance bounds for the matrix pencil method. Moreover we show that our methods can be interpreted as giving preconditioners for Vandermonde matrices, and we use this observation to design faster algorithms for super-resolution. We believe that these ideas may have other applications in designing faster algorithms for other basic tasks in signal processing.

\end{abstract}

\thispagestyle{empty}

\newpage

\setcounter{page}{1}

\section{Introduction}

\subsection{Background}

Super-resolution is a fundamental task in imaging, where the goal is to resolve features at a scale beyond the diffraction limit. This problem has received considerable attention both as a {\em statistical} and {\em algorithmic} problem \---- where the goal is to localize objects using only low frequency measurements \---- as well as an {\em engineering} problem \---- where the goal is to design devices that can work around the diffraction limit. In fact, the 2014 Nobel Prize in Chemistry was recently awarded to Eric Betzig, Stefan Hell and William Moerner 
\begin{quote}
\begin{center}
{\Large ``} ...for the development of super-resolved fluorescence microscopy... {\Large "}
\end{center}
\end{quote}
which has played a key role in a number of scientific discoveries, such as recent discoveries about how viruses change their shape before entering a cell \cite{R}. 

We will be interested in a popular mathematical abstraction of this problem that has been widely studied in the statistics, signal processing and machine learning communities \cite{D, CFG1, CFG2, BTR}, In particular, there is an unknown super-position of $k$ point sources: 
$$x(t) = \sum_{j = 1}^k u_j \delta_{f_j}(t)$$
where $\delta_{\tau}(t)$ is the Dirac delta function at $\tau$ and we assume that each $f_j \in [0, 1)$. Additionally each measurement $v_\ell$ of these point sources takes the form 
$$v_{\ell} = \int_{0}^1 e^{ i 2 \pi \ell t} x(t) \mbox{ } dt + \eta_\ell= \sum_{j=1}^k u_j e^{ i 2 \pi  f_j \ell} + \eta_\ell$$
where $\eta_\ell$ is noise in the measurement. However due to physical limitations, we are only able to measure the data at low-frequencies \---- i.e. we can observe $v_\ell$ for $|\ell | \leq m$ where $m$ is referred to as the {\em cutoff frequency}.  The basic question is, how large does $m$ need to be in order to recover the locations $f_j$ and the coefficients $u_j$? The problem of recovering fine-grained structure from coarse-grained measurements has a broad set of applications including in domains such as medical imaging, microscopy, astronomy, radar, geophysics and spectroscopy (see \cite{D, PPK, CFG1} and references therein). This is also a fundamental algorithmic and statistical problem in its own right. 

This particular estimation problem has a long history of study, and there are a variety of methods that work in the noise-free case for $m = k$ such as Prony's method (1795) \cite{dP}, Pisarenko's method (1973) \cite{Pi}, the matrix pencil method (1990) \cite{HS} and other linear prediction methods (see references in \cite{Stoi}). This is optimal since if we had any fewer than $2k + 1$ total measurements regardless of their frequency, the true signal would not necessarily be the sparsest solution to the inverse problem \cite{E}. Here we study this problem in the presence of noise, and establish that it 
exhibits a sharp phase transition. Let $\Delta$ denote the minimum separation between any pair of the $f_j$'s according to the wrap-around metric on the unit interval. Our main algorithmic result shows that if $m > 1/\Delta + 1$ then noisy super-resolution is possible and our estimates converge to the true values at an inverse polynomial rate in terms of the magnitude of the noise. This is tight in that if $m = (1-\epsilon)/\Delta$ there is a pair of $k$ point sources $x$ and $x'$ each with separation $\Delta$ where we would need the noise to be exponentially small to tell them apart. 

Along the way, we also prove a sharp phase transition for the condition number of the Vandermonde matrix based on connections to certain extremal functions studied in analytic number theory. To elaborate on this, one can think of a hypothesis testing version of various learning problems. Then in many settings we need a family of test functions that are strong enough to distinguish any pair of hypothesis with suitably different parameters. For mixtures of Gaussians, low degree polynomials work \cite{KMV}, \cite{MV2}, \cite{BS2} but we can think of the present paper as fitting into a broader agenda where the goal is to find more creative choices for these test functions which in turn lead to better algorithms (see also \cite{MS} where a similar framework is used to design algorithms for the population recovery problem).

\subsection{Previous Work and Connections}

Recall that there are a variety of algorithms for super-resolution that work for $m = k$ in the noise-free case. We review one popular technique called the {\em matrix pencil method} in Section~\ref{sec:mpm} where the idea is to express the unknown parameters as the solution to a generalized eigenvalue problem. There have been numerous attempts to analyze the noise tolerance of these methods, but all of them pass to limiting arguments as $m$ tends to infinity and yield ineffective bounds. The main bottleneck is that this approach relies on solving linear systems in Vandermonde matrices, which can be very poorly conditioned. In fact, many other algorithms for inverse problems that use closely related techniques work in the exact case but are either hopeless or require very different settings of their parameters in the presence of noise. For example:

\begin{itemize}

\item[(a)] Given the evaluation of a degree $k$ polynomial $p(x)$ at $k + 1$ distinct points we can solve for the coefficients of $p(x)$ and then we can evaluate $p(x)$ anywhere. But if we are given the values of $p(x)$ up to some additive noise $\eta$ even for {\em all} values of $x \in [0, 1]$ we can only hope to estimate $p(2)$ within an additive $c(2 + \sqrt{3})^k \eta$. In fact this is tight by standard results in approximation theory \cite{C}. 

\item[(b)] In sparse recovery, there is an unknown $s$-sparse signal $x$ and we observe $A x + \eta$. If $\eta = 0$ and $A$ is chosen to be the first $2s$ rows of the discrete Fourier transform matrix we can recover $x$ exactly using Prony's method. But $A$ must have at least $\Omega( s \log n/s)$ rows in order to recover $x$ in the presence of some moderate amount of noise \cite{DIPW}.

\end{itemize} 

\noindent We could wonder whether super-resolution suffers from the same type of extreme noise-intolerance that makes it possible in the exact case, but requires (say) $\eta_\ell$ to be {\em exponentially} small in order to recover the $u_j$'s and $f_j$'s to any reasonable accuracy.

Prior to the more modern work on super-resolution, there was a considerable effort to develop empirical methods that work when there is some a priori information about the sparsity of the signal (see references in \cite{D}). In a remarkable work, Donoho formulated the notion of the {\em Rayleigh index} which can be thought of as a continuous analogue to separation, and studied super-resolution where the $f_j$'s are restricted to be on a grid \cite{D}. Donoho used Beurling's theory of interpolation \cite{B} to prove that any $\Delta$-separated signal can be uniquely recovered from from its Fourier transform on the continuous range $[-1/\Delta, 1/\Delta]$. However this result is information theoretic in the sense that there could be two distinct signals that are both $\Delta$-separated, but whose Fourier transforms restricted to $[-1/\Delta, 1/\Delta]$ differ in a measure zero set. (Hence they are not noticeably different). Donoho also gave results on the modulus of continuity for the recovery of $\Delta$-separated signals, but these results lose an extra factor of two in the range of low-frequency measurements that are needed. The main question left open by this work was to prove recovery bounds when the signal is not restricted to be on a grid.

Subsequently, Candes and Fernandez-Granda \cite{CFG1, CFG2} proposed an exciting new method for solving this inverse problem based on convex programming. 
The authors proved a variety of results in the setting where (a) there is no noise or (b) there is noise, but the $f_j$'s are restricted to be on a grid. 
In the noise-free setting, the authors showed
their approach recovers the $u_j$'s and $f_j$'s exactly provided that $m \geq \frac{2}{\Delta}$ and $m \geq 128$ where $\Delta$ is the minimum separation. The authors further improved this bound when the $u_j$'s are real to a separation condition of $m \geq \frac{1.87}{\Delta}$ (unpublished)\footnotemark[1]. In the grid setting, the authors defined a parameter $C$ called the {\em super-resolution factor} and proved various bounds using it on the error of their estimator. However the approach used to measure the error is non-standard since it uses the Fejer kernel to `mask' the error on the high-frequency parts of the signal. Finally, Ferndanez-Granda \cite{F} returned to the setting where the $f_j$'s can be at arbitrary locations, and showed that if the noise is small enough the estimate becomes localized around the true spikes. After the initial publication of our work, we became aware of the recent work of Liao and Fannjiang \cite{LF} who analyze the MUSIC algorithm. Their approach is also based on upper-bounding the condition number of the Vandermonde matrix and works with a separation condition of $m \geq (1+C(\Delta))/\Delta$ where $C(\Delta)$ goes to zero as $\Delta$ increases. We will prove a sharper bound that is optimal for all $\Delta$, through extremal functions. 

\footnotetext[1]{One should note though, that the methods we described above work in the noise-free setting for $m = k$ without any separation condition. So the crucial issue really is in understanding when super-resolution is possible in the presence of noise.}

\paragraph{Further Connections} In fact, there are a number of connections  \---- both conceptual and technical \---- to the rich literature on the sparse Fourier transform (SFFT). We can interpret\footnotemark[2] recent works as providing algorithms for approximately recovering a signal that is the super-position of $k$ point sources at discrete locations, even in the presence of noise, from few measurements \cite{GGIMS, GMS} and state-of-the-art algorithms run in time $O(k \log k \log n/k)$ \cite{HIKP}. What if these point sources are not at discrete locations and are allowed to be {\em off the grid}? In fact, Bhaskar et al. \cite{BTR} defined a problem called {\em compressed sensing off the grid} where we are given a subset of the low frequency measurements. More precisely, we are given $v_\ell$ for all $\ell \in S \subset \{-m, -m+1, ..., m-1, m\}$ where $|S|$ is roughly $O(k \log m/k)$. The authors proved that if $S$ is chosen uniformly at random and $m \geq 2/\Delta$, the convex program recovers the $u_j$'s and $f_j$'s exactly. Additionally, the authors needed a technical assumption that each $u_j \in \{\pm 1\}$ and is chosen uniformly at random. However these algorithms are only known to work in the noise-free setting, and it remains an interesting open question whether compressed sensing off the grid can be solved in the presence of noise, with similar guarantees to what we are able to obtain here.

\footnotetext[2]{The guarantees in these works are even stronger, since the usual goal is to recover a $k$-sparse approximation to the discrete Fourier transform (DFT) whose error is within a $(1+\epsilon)$ factor of the error of the {\em best} $k$-sparse approximation. In particular, these algorithms work even when the signal is not particularly close to a sum of $k$ sinusoids.}

\subsection{Our Results and Techniques}


Here we show the first noise tolerance bounds for the {\em matrix pencil method} (see Section~\ref{sec:mpm}) and we use this to give algorithms for noisy super-resolution with the {\em optimal} relationship between the number of measurements and the separation. 
Our key technical ingredient is giving a tight characterization of when the Vandermonde matrix $V_m^k(\alpha_1, \alpha_2, ..., \alpha_k)$ (see Section~\ref{sec:mpm}) is well conditioned. Set $\alpha_j = e^{i 2 \pi f_j }$. Recall that $\Delta$ is the minimum separation of the $f_j$'s according to the wrap-around metric on the unit interval. 

\begin{theorem}
The condition number $\kappa$ of $V_m^k(\alpha_1, \alpha_2, ..., \alpha_k)$ satisfies $\kappa^2 \leq \frac{m + 1/\Delta - 1}{m - 1/\Delta -1}$ provided that $m > \frac{1}{\Delta} + 1$. 
\end{theorem}

This theorem is easy to prove, given the right extremal functions. The key is the Beurling-Selberg majorant and minorant\footnotemark[3] (see Section~\ref{sec:beur}) which majorize and minorize the sign function on the real line but have a compactly supported Fourier transform \cite{Se}. In fact, these functions are the solutions to their respective extremal problems and have led to a sharp forms of various inequalities in analytic number theory. Most notably these functions can be used to give simple, direct proofs of the large sieve inequality (see \cite{V}) and extensions of Hilbert's inequality \cite{MV}. These inequalities yield strong upper and lower bounds for integrals of sums of exponentials, and here we adapt the techniques to our setting to reason about evaluating sums of exponentials at discrete sets of points. This readily yields our upper bound on the condition number. 

\footnotetext[3]{In fact, the minorant is the more important of the two extremal functions for our purposes. It is the opposite for most of the standard applications, where the majorant plays the key role and the minorant is a footnote.}

Indeed it is the condition number of $V_m^k$ that governs whether or not the matrix pencil method is stable. Despite its widespread use in signal processing, this is the first {\em effective} bound on its noise tolerance and we get immediate consequences for super-resolution. Throughout this paper we will quantify the error between our estimates $\{ (\widehat{u}_j, \widehat{f}_j)\}$ and the true values $\{(u_j, f_j)\}$ based on the best matching between them:
$$\min_{\pi} \max_j \max\Big( \| \widehat{u}_j - u_{\pi(j)}\|, \|\widehat{f}_j - f_{\pi(j)}\| \Big)$$
where $\pi$ is a permutation. In particular, we refer to our estimates as being $\epsilon$-close to the true parameters if there is a matching of the $k$ spikes in our estimate to the true $k$ spikes so that across the matching we get both the coefficients and the locations correct within an additive $\epsilon$. 

\begin{theorem}
If the cut-off frequency satisfies $m > 1/\Delta + 1$ where $\Delta$ is the minimum separation (according to the wrap-around metric), then there is a polynomial time algorithm to recover the $u_j$'s and $f_j$'s whose estimates converge to the true values at an inverse polynomial rate in terms of the magnitude of the noise.
\end{theorem}

\noindent We will state the precise quantitative bounds later (see Theorem~\ref{thm:withnoise}). These results are tight (as we will describe next) and settle the question of what is the best relationship between the number of measurements and the minimum separation that we can achieve for super-resolution. Notice that this theorem has none of the restrictions that were needed in some of the previous work; there are no restrictions on the $f_j$'s being on a grid, or on the coefficients. 

As we noted earlier, if $k = 1/\Delta$ and $m < 1/\Delta$ then the true signal $x$ is not necessarily the sparsest solution to the inverse problem. This is related to the fact that $V_m^k$ is not full rank. However it is easy to give more compelling lower bounds where $k$ is itself much smaller than $m$. Candes and Fernandez-Granda \cite{CFG1} proved almost what we need, and gave an $x$ that is $\Delta$-separated where $\int_E \|\widehat{x}(t)\|^2 \mbox{ } dt < 2^{-\Omega(\epsilon k)}$ for $E = [-m/2, m/2]$ and $m = (1-\epsilon)/\Delta$. Their approach was based on the asymptotics of Slepian's discrete prolate spheroidal wave functions which were developed in a sequence of papers \cite{Sl1, Sl2, Sl3, Sl4, Sl5}. Here we give a discrete version of the above result, and give an elementary and self-contained proof. First, we prove:



\begin{theorem}\label{thm:need}
If $m < (1-\epsilon)/\Delta$ and $k = \Omega(\log m)$, the condition number of $V_{m}^k(\alpha_{k/2}, \alpha_{k/2 + 1}, ..., \alpha_{3k/2 -1})$ is at least $2^{\Omega(\epsilon k)}$.
\end{theorem}

\noindent It turns out that this gives us a pair of point sources that are each $2\Delta$-separated, where setting the cut-off frequency $m = (1-\epsilon)/2\Delta$ we would need the magnitude of the noise to be exponentially small in order to tell them apart. Hence we cannot even solve the following hypothesis testing variant in the presence of noise:

\begin{corollary}[Informal]
If the cut-off frequency satisfies $m < (1-\epsilon)/\Delta$ where $\Delta$ is the minimum separation (according to the wrap-around metric), then there is a pair $x$ and $x'$ of $k$ point sources each with separation $\Delta$ where we would need $|\eta_\ell| \leq 2^{-\Omega(\epsilon k)}$ in order to tell $x$ and $x'$ apart.
\end{corollary}

\noindent This proves the optimality of our algorithms in a strong sense. In fact, this impossibility result holds even in a stochastic model, where the noise $\eta_\ell$ has its real and complex part chosen as Gaussian independent random variables with variance $2^{-\Omega(\epsilon k)}$. We state the precise result later (see Corollary~\ref{corr:impossible}). 


We can think about our use of majorants and minorants as being a ``two-function method" where we upper and lower bound some function we are interested in \---- in this case $\|V_m^k u\|_2^2$. 
We develop another framework which we call a ``one-function method" in which we prove upper bounds on the condition number of $V_m^k$ in an alternate manner. For example, the proof of the Salem inequality (see Section~\ref{sec:salem}) relies on lower bounding some integral we are interested in as: $$\int_{-\infty}^{\infty} \chi(t) \|f(t)\|^2 \mbox{ } dt  \leq \int_E \|f(t)\|^2 \mbox{ } dt $$ where $\chi(t)$ is an appropriately chosen ``test function" that is supported on the interval $E$. If additionally $\widehat{f}$ is well-spaced and $\widehat{\chi}$ decays quickly, then we can write the left hand side as a relatively simple sum-of-squares plus some correction terms. These same sorts of techniques have been applied in the sparse Fourier transform literature, where a combination of {\em hashing} and {\em filtering} is used to identify large coefficients \cite{GGIMS, GMS, HIKP}. 

Our key observation is that if we transfer these types of inequalities to the discrete setting, the test function itself can be thought of as a way to re-weight the rows of $V_m^k$ to precondition it. This is a very different setting than usual because we do not know $V_m^k = V_m^k(\alpha_1, \alpha_2, ..., \alpha_k)$ \---- the matrix we want to precondition \---- only a family that it belongs to. We call this {\em universal preconditioning}, which seems to be an interesting notion in its own right. The key point is that $\chi(t)$ is an universal preconditioner for any $V_m^k$ where the $f_j$'s are well-separated, and indeed we will use it to design faster algorithms for super-resolution. Here we restrict to the case where each $u_j =1$ to keep things simple. Our final result is:

\begin{theorem}\label{thm:last}
If the cut-off frequency satisfies $m > \Omega( (1/\Delta) \log 1/\epsilon)$ where $\Delta$ is the minimum separation (according to the wrap-around metric) and each $u_j = 1$, there is an $\widetilde{O}(m^2)$ time algorithm to recover the $f_j'$'s up to accuracy $\epsilon$, provided that the noise satisfies $\|\eta_j\| \leq \epsilon^2/(4k)$ for each $j$. 
\end{theorem}

\noindent We prove this result by using $\chi(t)$ to construct a function $F(z)$ whose local minima are $\epsilon$-close to the $f_j$'s. Moreover we establish that $F(z)$ is approximately strongly convex in a neighborhood around each $f_j$, and we give a variant of gradient descent that finds all of them in parallel. 
This approach is much faster than trying to solve the convex program in \cite{CFG1} or finding {\em all} the generalized eigenvalues for a pair of $m \times m$ Toeplitz matrices as in the matrix pencil method. 
We remark that least-squares is a popular approach to solve super-resolution, but most analyses  pass to a limiting argument and use the fact that as $m$ goes to infinity, the (re-normalized) columns of $V_m^k$ become orthogonal. Our approach instead is to precondition $V_m^k$ and we believe that this idea of using test functions (i.e. filters, kernels) as a preconditioner may have other applications in designing faster algorithms for other basic tasks in signal processing.

\section{The Optimal Separation Criteria}

\subsection{The Matrix Pencil Method}\label{sec:mpm}

Here we review the matrix pencil method, which is a basic tool in statistics, signal processing and learning. The basic idea is to set up a generalized eigenvalue problem whose solutions are the unknown parameters. In particular, in the setting of super-resolution we are given 
$v_\ell = \sum_j u_j e^{i 2 \pi  f_j \ell} + \eta_\ell$ for integers $|\ell| \leq m$ and where $\eta_\ell$ is noise. We will limit our discussion (for now) to the idealized case to explain the method, and later we will see that it is the condition number of the Vandermonde matrix $V_m^k$ (defined below) that controls whether or not this method is noise tolerant. Let $\alpha_j = e^{- 2 \pi i f_j }$, then:
$$V_m^k(\alpha_1, \alpha_2, ..., \alpha_k) = \left [ \begin{array}{cccc} 
									1 & 1 & \dots & 1 \\
									\alpha_1 & \alpha_2 & \dots & \alpha_k \\
									\vdots & \vdots & \dots & \vdots \\
									\alpha_1^{m-1} & \alpha_2^{m-1} & \dots & \alpha_k^{m-1} 
									\end{array} \right ]$$
Moreover let $D_u = \mbox{diag}(\{u_j\})$ and let $D_\alpha = \mbox{diag}(\{\alpha_j\})$. Recall that in a generalized eigenvalue problem we are given a pair of matrices $A$ and $B$ and the goal is to find the generalized eigenvalues $\lambda$ for which there is a vector $x$ that satisfies $Ax = \lambda B x$. In particular, the eigenvalues of $A$ are the solution to the generalized eigenvalue problem where $B = I$. 

Then the following well-known observations are the key to the matrix pencil method:

\begin{observation}
The non-zero generalized eigenvalues of $V D_u V^H$ and $V D_u D_\alpha V^H$ are exactly $\{\alpha_j\}$.
\end{observation}

\begin{observation}
The entry in row $i$, column $j$ of $V D_u V^H$  is $v_{i -j}$ and the corresponding entry in $V D_u D_\alpha V^H$ is $v_{i - j + 1}$
\end{observation}

\noindent Hence the entries in these matrices are precisely the low-frequency measurements we have access to. Now it is straightforward to find the unknown frequencies $\{f_j\}$ by setting up the above generalized eigenvalue problem and solving it. Moreover we can find the coefficients too by solving the following linear system: $$v = V_n^k(\alpha_1, \alpha_2, ..., \alpha_k) u$$ where $v$ and $u$ are the vectors of $\{v_\ell\}$'s and $\{u_j\}$'s respectively. 

This approach succeeds provided that $V_m^k(\alpha_1, \alpha_2, ..., \alpha_k)$ has full column rank, which is true if the $f_j$'s are distinct and it has at least as many rows as columns. Thus as long as the cutoff frequency satisfies $m \geq k +1 $ we can solve the super-resolution problem exactly in the noise-free setting. This algorithmic guarantee has been rediscovered numerous times. However this approach can sometimes be extremely intolerant to noise. What governs the noise-tolerance of the above generalized eigenvalue problem (and consequently the overall method) is the condition number of $V_m^k(\alpha_1, \alpha_2, ..., \alpha_k) $. There are by now many known perturbation bounds for generalized eigenvalue problems (see \cite{SS}), however we must take extra care in that the rectangular Vandermonde matrix will have full column rank but will not be invertible in our setting. We defer proving perturbation bounds for the above problem until Section~\ref{sec:genep}, but let us for now focus on understanding when the Vandermonde matrix is well-conditioned. 

\subsection{The Beurling-Selberg Majorant}\label{sec:beur}

The Beurling-Selberg majorant $B(t)$ plays a key role in analytic number theory, where it has been used to establish sharp forms of the large sieve inequality and also important extensions to the classic Hilbert inequality. In particular:

\begin{definition}
$B(t) = \Big ( \frac{\sin \pi t}{\pi} \Big )^2 \Big ( \sum_{j=0}^\infty (t - j)^{-2} - \sum_{j = -\infty}^{-1} (t - j)^{-2}  + 2 t^{-1} \Big )$
\end{definition}

We will describe its key properties next but for now we remark that $B(t)$ is an entire function of exponential type $2 \pi$ \-- i.e. its Fourier transform $\widehat{B}$ is supported on $[-1, 1]$. Indeed, the definition above looks quite strange at first glance but any function of exponential type $2\pi$ satisfies a certain interpolation formula based on its values at the integers (see e.g. \cite{T} or \cite{Z}) and $B(t)$ comes from plugging in the sign function instead (which is not of exponential type $2 \pi$). 

The key properties of $B(t)$ are that it majorizes the sign function $\mbox{sgn}$ \---- i.e. $\mbox{sgn}(t) \leq B(t)$ \---- and that:
$$\int_{-\infty}^{\infty} B(t) - \mbox{sgn}(t) \mbox{ } dt = 1$$
\noindent In fact, $B(t)$ is the extremal function that satisfies the above properties. Hence any function $F(t)$ of exponential type $2 \pi$ that majorizes the sign function satisfies $$\int_{-\infty}^{\infty}F(t) - \mbox{sgn}(t) \mbox{ } dt \geq 1$$ and equality holds if and only if $F(t) = B(t)$. There is a similar construction that minorizes the sign function, and is also extremal. However these functions are usually used to majorize and minorize a given interval, and we will instead state the theorems we need restricted to that setting.

Let $E = [-n, +n]$ and let $I_E(t)$ be the indicator function for $E$. Then

\begin{theorem}
There are entire functions $C_E(t)$ and $c_E(t)$ whose Fourier transform is supported on $[-\Delta, \Delta]$ and that satisfy $c_E(t) \leq I_E(t) \leq C_E(t)$ and $$\int_{-\infty}^{\infty}C_E(t) - I_E(t) \mbox{ } dt = \int_{-\infty}^{\infty} I_E(t) - c_E(t) \mbox{ } dt = 1/\Delta$$
\end{theorem}

\noindent The functions $C_E(t)$ and $c_E(t)$ are called majorants and minorants respectively. In \cite{MV}, Montgomery and Vaughan used these properties of $C_E(t)$ and $c_E(t)$ to give a strong generalization of Hilbert's inequality. Here, we will modify their approach and prove upper bounds on the condition number of the Vandermonde matrix. 

\subsection{Upper Bounding the Condition Number}

Here we will adapt the approach of Montgomery and Vaughan \cite{MV} to give an upper bound on the condition number of $V_m^k$. Recall that $\alpha_j = e^{i 2 \pi  f_j }$ and $d_w$ is the wrap-around distance on the unit interval. Then suppose each $f_j \in [0, 1)$ and the minimum separation is strictly larger than $\Delta$ \---- i.e. $d_w(f_j, f_{j'}) > \Delta$ for $j \neq j'$. We prove that  $V_m^k(\alpha_1, \alpha_2, ..., \alpha_k)$ is well-conditioned provided that $m > \frac{1}{\Delta} + 1$. Moreover we will prove in Section~\ref{sec:need} that if $m \leq \frac{1 - \epsilon}{\Delta}$ then  $V_m^k$ can have condition number that is exponentially large in $k$, so there is a sharp transition. 

Let $\widehat{C}_E$ denote the Fourier transform  of $C_E$ and let $h(\ell) = \sum_{\ell = -\infty}^\infty \delta(\ell) = \sum_{t = -\infty}^{\infty} e^{i 2 \pi t \ell} $ denote the Dirac comb, where we have used the fact that the Dirac comb is its own Fourier transform. Recall that $E = [-n, +n]$, $I_E(t)$ is the indicator function for $E$ and $v_{\ell} = \sum_{j = 1}^k u_j e^{i 2 \pi  f_j \ell}$. Then we can write
\begin{eqnarray*}
\sum_{\ell = -n}^n \|v_{\ell}\|^2 &=& \int_{-\infty}^\infty h(\ell) I_E(\ell) \|v_{\ell}\|^2 \mbox{ } d\ell \leq \int_{-\infty}^\infty h(\ell) C_E(\ell) \|v_{\ell}\|^2 \mbox{ } d\ell \\
&=& \sum_{j = 1}^k \sum_{j' = 1}^k u_j \bar{u}_{j'} \int_{-\infty}^\infty h(\ell) C_E(\ell) e^{i 2 \pi (f_j - f_{j'}) \ell}\mbox{ } d\ell \\
&=& \sum_{j = 1}^k \sum_{j' = 1}^k \sum_{t = -\infty}^{\infty}  u_j \bar{u}_{j'} \int_{-\infty}^\infty e^{i 2 \pi t \ell} C_E(\ell) e^{i 2 \pi (f_j - f_{j'}) \ell}\mbox{ } d\ell \\
&=& \sum_{j = 1}^k \sum_{j' = 1}^k \sum_{t = -\infty}^{\infty} u_j \bar{u}_{j'}  \widehat{C}_E(f_j - f_{j'} + t) =  \widehat{C}_E(0) \sum_{j = 1}^k |u_j|^2
\end{eqnarray*}
\noindent where the last equality follows because $|f_j - f_{j'} + t |> \Delta$ for any integer $t$, and because $\widehat{C}_E$ is supported on $[-\Delta, \Delta]$. Moreover $$\widehat{C}_E(0) = \int_{-\infty}^{\infty}C_E(\ell) \mbox{ } d\ell = |E| + 1/\Delta = 2n + 1/\Delta$$ We can instead use the minorant $c_E(\ell)$ and combined with the above inequality we get $$(2n -1/\Delta) \sum_{j = 1}^k |u_j|^2 \leq \sum_{\ell = -n}^n \|v_{\ell}\|^2 \leq (2n + 1/\Delta) \sum_{j = 1}^k |u_j|^2$$

Now we can make the substitution $b_j = u_j e^{- i 2 \pi  f_j n}$ in which case $$\|V_{2n + 1}^k(\alpha_1, \alpha_2, ..., \alpha_k) b\|^2 = \sum_{\ell = -n}^n \|v_{\ell}\|^2 = (2n \pm 1/\Delta) \sum_{j = 1}^k \|b_j\|^2$$ Recall that the condition number $\kappa(V) = \frac{\max_a \|V a\|}{\min_{a'} \|V a'\|}$ and hence we obtain the following theorem by setting $m = 2n + 1$:

\begin{theorem}
The condition number $\kappa$ of $V_m^k(\alpha_1, \alpha_2, ..., \alpha_k)$ satisfies $\kappa^2 \leq \frac{m + 1/\Delta - 1}{m - 1/\Delta -1}$ provided that $m > \frac{1}{\Delta} + 1$. 
\end{theorem}

\subsection{The Generalized Eigenvalue Problem}\label{sec:genep}

There are a number of well-known results that bound the noise-tolerance of a generalized eigenvalue problem (see \cite{SS} and references therein). However often these results assume that the matrices $A$ and $B$ are both full rank. We will be able to reduce to this case and appeal to their perturbation bounds as a black-box. For convenience, we will adopt the same notation as in \cite{SS} and let us first focus on settings where $A$ and $B$ are square matrices of order $k$. We are interested in the solutions to $\mbox{det}(\alpha A - \beta B) = 0$. 

Additionally let $\lambda = \alpha/\beta$. Then the pair $(A, B)$ is called {\em regular} if there is any choice of $\alpha$ and $\beta$ for which $\mbox{det}(\alpha A - \beta B) \neq 0$. Of course, if either $A$ or $B$ is full rank then the pair is regular. The bounds in \cite{SS} are stated in terms of the chordal metric, which will be convenient for our purposes.

\begin{definition}
The chordal metric for $\lambda, \mu \in \mathbb{C}$ is defined as $$\chi(\lambda, \mu) = \frac{ |\lambda - \mu|}{\sqrt{1 + |\lambda|^2} \sqrt{1 + |\mu|^2}}$$
\end{definition}

\noindent Equivalently if we take $s(\lambda)$ and $s(\mu)$ to be the points on $\mathbb{S}^2$ so that $\lambda$ and $\mu$ are their stereographic projections, then $\chi(\lambda, \mu)  = \frac{1}{2} \|s(\lambda) - s(\mu)\|_2$. This metric is well-suited for our applications since our eigenvalues $\lambda$ will anyways be on the boundary of the complex disk. 

Throughout this section let $\widehat{A} = A + E$ and let $\widehat{B} = B + F$. Let $\widehat{\lambda}$ be the generalized eigenvalues for the pair $\widehat{A}, \widehat{B}$. We will work with the following measure:

\begin{definition}
The matching distance with respect to $\chi$ is defined as $$\mbox{md}_{\chi}[(A, B), (\widehat{A}, \widehat{B})] = \min_\pi \max_{i} \chi(\lambda_i, \widehat{\lambda}_{\pi(i)})$$ where $\pi$ is a permutation on $\{1, 2, ..., k\}$. 
\end{definition}

\noindent Note that we can define the matching distance analogously using another metric. In fact, we will state the guarantees of the matrix pencil method based on the matching distance of the output $\widehat{f}_j$ to the true $f_j$'s according to the wrap-around metric. 

We will make use of the following theorem, which combines Theorem $2.4$ and Corollary $2.5$ from \cite{SS}:

\begin{theorem}\label{thm:gersh2} \cite{SS}
Let $(A, B)$ and $(\widehat{A}, \widehat{B})$ be regular pairs and further suppose that for some nonsingular $X$ we have $(X A X^H, X B X^H) = (I, D)$ where $D$ is diagonal. Also let $e_i$ be the $i$th row of $X(A - \widehat{A})X^H$ and  let $f_i$ be the $i$th row of $X (B-\widehat{B}) X^H$ and set $\rho = \max_i \{\| e_i \|_1 + \| f_i \|_1\}$. If the following regions 
$$\mathcal{G}_i = \Big\{ \mu \Big| \chi(D_{ii}, \mu) \leq \rho\Big\}$$
are disjoint then the matching distance of the generalized eigenvalues of $(\widehat{A}, \widehat{B})$ to $\{D_{ii}\}$ is at most $\rho$ with respect to the chordal metric. 
\end{theorem}

\subsection{A Modified Matrix Pencil Method}

Now we are ready to prove error bounds for the matrix pencil method. Recall that we are interested in solving the generalized eigenvalue problem $(V D_u V^H, V D_u D_\alpha V^H)$ where $V$ is short-hand for $V_m^k$, $D_u = \mbox{diag}(u_j)$ and $D_\alpha = \mbox{diag}(\alpha_j)$. Instead we are given $VD_uV^H + E$ and $VD_u D_\alpha V^H + F$. We will first reduce to a generalized eigenvalue problem on order $k$ matrices so that the matrix pair will be regular. Ideally, we would project onto the column span of $V$ to preserve the fact that the $\alpha_j$'s are the generalized eigenvalues. However we do not have access to this space, and instead we will project onto the top $k$ singular vectors of $VD_uV^H + E$.  Let $\widehat{U}$ be an orthonormal basis for this space. 

It follows from Wedin's theorem (see e.g. \cite{HJ}) that the column space of $V$ and the span of $\widehat{U}$ have a small principal angle. In fact, an even stronger statement is true: we can align the basis of the column span of $V$ to $\widehat{U}$. Throughout this section, let $u_{min} = \min_j |u_j|$ and $u_{max} = \max_j |u_j|$. 

\begin{lemma}
There is an orthonormal basis $U$ for the column span of $V$ that satisfies $\|U - \widehat{U}\|_2 \leq  \frac{2 \|E\|_2}{u_{min}} $. 
\end{lemma}

\begin{proof}
By Wedin's theorem and the definition of the principal angle we have that $$\sin( \Theta(U, \widehat{U})) = \|(I- U U^H) \widehat{U} \|_2 \leq \frac{\|E\|_2}{u_{min}}$$ Then using Weyl's inequality (see e.g. \cite{HJ}) we can write 
$$\|U - \widehat{U}\|_2 \leq \|U U^H \widehat{U} - \widehat{U}\|_2 + \|U - U U^H \widehat{U}\|_2 $$
The first term is just $\sin( \Theta(U, \widehat{U}))$. Also we can use the fact that $\cos( \Theta(U, \widehat{U})) = \sigma_{\min}(U^H \widehat{U})$. Then we can bound the second term as
$$ \|U - U U^H \widehat{U}\|_2  \leq \|I - U^H \widehat{U}\|_2 \leq 1 - \cos( \Theta(U, \widehat{U}))  \leq \sin( \Theta(U, \widehat{U})) \leq \frac{\|E\|_2}{u_{min}}$$
And this completes the proof. 
\end{proof}

Now let $A = U^H V D_u V^H U$, let $B = U^H V D_u D_{\alpha} V^H U$ and let $X = D_u^{-1/2} V^+ U$. It follows that $X A X^H= I$ and $X B X^H = D_{\alpha}$.  Also consider the perturbed generalized eigenvalue problem for $\widehat{A} = \widehat{U}^H (V D_u V^H + E) \widehat{U}$ and $\widehat{B} = \widehat{U}^H (V D_u D_\alpha V^H + F) \widehat{U}$. Let us set $\|U - \widehat{U}\|_2= \tau$ for notational convenience. 

Then we can write
$\|X(A - \widehat{A})X^H\|_2 \leq  \|X\|_2^2 \|A - \widehat{A}\|_2$.
And furthermore
$$ \|A - \widehat{A}\|_2 \leq \|E\|_2 + \|U^H V D_u V^H U - \widehat{U}^H V D_u V^H \widehat{U}\|_2 \leq \|E\|_2 + \Big(2 \tau \|V\|_2^2 + \tau^2 \|V\|_2^2\Big)u_{max}$$
And combining these inequalities and $\tau = \frac{2 \|E\|_2}{u_{min}} $ we have
$$\|X(A - \widehat{A})X^H\|_2 \leq \frac{\|E\|_2}{\sigma_{min}^2 u_{min}} + 4 \kappa^2 \Big( \frac{\|E\|_2}{u_{min}} + \frac{\|E\|_2^2}{u_{min}^2}\Big) \frac{u_{max}}{u_{min}} $$
where $\kappa$ and $\sigma_{\min}$ are the condition number and minimum singular value of $V$ respectively. 
An identical bound holds for $B$ and $\widehat{B}$ where we replace $E$ by $F$. We remark that a bound on the spectral norm of $X(A - \widehat{A})X^H$ is also a bound on the $\ell_2$-norm of any of its rows; and to get a bound on the $\ell_1$-norm of any its rows, we can multiply the above bound by  $k^{1/2}$. 

Hence we can combine the above bounds and Theorem~\ref{thm:gersh2} to get that the generalized eigenvalues of $(\widehat{A}, \widehat{B})$ converge to the generalized eigenvalues of $(A, B)$ at an inverse polynomial rate with the noise, with respect to the chordal metric. We remark that the generalized eigenvalues of $(A, B)$ are on the complex disk. Hence if $(\widehat{\alpha}, \widehat{\beta})$ is a generalized eigenvalue of $(\widehat{A}, \widehat{B})$, we can find the closest point on the complex disk to $\widehat{\lambda} = \widehat{\alpha}/\widehat{\beta}$ which in turn will be have arc-length at most $O(\chi(\widehat{\lambda}, \lambda))$ where $\lambda = \alpha/\beta$ and $(\alpha, \beta)$ is the generalized eigenvalue of $(A, B)$ that is matched to  $(\widehat{\alpha}, \widehat{\beta})$. 

Hence the output of {\sc ModifiedMPM} is a set of locations $\{\widehat{f}_j\}$ whose matching distance to $\{f_j\}$ with respect to $d_w$ converges at an inverse polynomial rate to zero with the magnitude of the error. All that remains is to recover the coefficients too. To this end, we write $\widehat{V}$ to denote the Vandermonde matrix with $\widehat{f}_j$'s instead of $f_j$'s. Then we are given 
$v= V u + \eta = \widehat{V} u + (V - \widehat{V})u + \eta$.
Then if we solve the linear system $v = \widehat{V} \widehat{u}$ we get $\widehat{u}$ such that 
$$\|u - \widehat{u}\|_2 \leq \frac{\|V - \widehat{V}\|_2 \|u\|_2}{\|\eta\|_2}$$
Let the matching distance between $\{f_j\}$ and $\{\widehat{f}_j\}$ be $\rho$. Then it is easy to see that each entry in $V - \widehat{V}$ is bounded by $O(\rho m)$ and hence
$$\|u - \widehat{u}\|_2 \leq O\Big( \frac{ \rho m^{3/2} k u_{max} + \|\eta\|_2}{m -1 - \frac{1}{\Delta - 2\rho}} \Big)$$
since the minimum separation of $\{\widehat{f}_j\}$ is at least $\Delta - 2\rho$. 

\begin{fragment*}[t]
\caption{
\label{alg:iterrefine}{\sc ModifiedMPM},  \textbf{Input:} $k$ and estimates $v_\ell$ for $|\ell | \leq m$ \\   \textbf{Output:} $\{(\widehat{u}_j, \widehat{f}_j)\}$  \vspace*{0.01in}
}

\begin{enumerate} \itemsep 0pt
\small 
\item Let $\widetilde{A}$ and $\widetilde{B}$ be matrices where $\widetilde{A}_{i,j} = v_{i - j}$ and $\widetilde{B}_{i,j} = v_{i - j + 1}$
\item Let $\widehat{U}$ be the top $k$ singular vectors of $\widetilde{A}$
\item Set $\widehat{A} = \widehat{U}^H \widetilde{A} \widehat{U}$ and $\widehat{B} = \widehat{U}^H \widetilde{B} \widehat{U}$
\item Solve for the generalized eigenvalues of $(\widehat{A}, \widehat{B})$
\item For each generalized eigenvalue $\widehat{\lambda} = \widehat{\alpha}/\widehat{\beta}$, let $\widehat{f}_j$ be the argument of the projection of $\widehat{\lambda}$ onto the complex disk
\item Set $\widehat{V}$ to be the Vandermonde matrix with $\{\widehat{f}_j\}$
\item Solve for $\widehat{u}$ such that $\widehat{V} \widehat{u} = v$
\item Output $\{(\widehat{u}_j, \widehat{f}_j)\}$
\end{enumerate} 

\end{fragment*}

We can now state our main theorem, which gives a quantitative rate for how quickly our estimates converge to the true values:

\begin{theorem}\label{thm:withnoise}
Let $\eta$ be the vector with entries $\{\eta_\ell\}$ and set
$$\gamma = \frac{ k \|\eta\|_2}{\sigma_{min}^2 u_{min}} + 4 \kappa^2 \Big( k \frac{\|\eta\|_2}{u_{min}} + k^{3/2} \frac{\|\eta\|_2^2}{u_{min}^2} \Big) \frac{u_{max}}{u_{min}} \mbox{ and } \zeta = O\Big( \frac{ \gamma m^{3/2} k u_{max} + \|\eta\|_2}{m -1 - \frac{1}{\Delta - 2\gamma}} \Big)$$ 
Further suppose that $\|E\|_2 + \|F\|_2 < \sigma_{\min}^2 u_{min} $ and $\gamma < \Delta/4$. 
Then the output of {\sc ModifiedMPM} satisfies
$$d_w(f_j, \widehat{f}_{\pi(j)}) \leq 2\gamma \mbox{ and } |u_j - \widehat{u}_{\pi(j)}| \leq \zeta$$
for all $j$, for some permutation $\pi$. 
\end{theorem}

\noindent The important point is that the estimates converge at an inverse polynomial rate to the true values as a function of $\|\eta\|_2$ provided that $m > (1+\epsilon)/\Delta$ for any $\epsilon > 0$. 




\section{The Need for Separation}\label{sec:need}

In the previous section, we gave an upper bound for the condition number of $V_m^k$ based on $m$ and $\Delta$. Here we establish that there is a sharp phase transition, and if $m = (1-\epsilon)/\Delta$ then $V_m^k$ has an exponentially large condition number. From this it follows that if $m = (1-\epsilon)/\Delta$ there is a pair of of $k$ point sources $x$ and $x'$ each with separation $\Delta$ which we would need exponentially small noise to tell apart. 

\subsection{Impossibility Results}

In order to make the connection to super-resolution it will actually be more convenient to work with the shifted Vandermonde matrix (below), but all the results in this section carry-over the other case as well.

$$V_{-n/2, n/2}^k(\alpha_1, \alpha_2, ..., \alpha_k) = \left [ \begin{array}{cccc} 
									\alpha_1^{-n/2} & \alpha_2^{-n/2} & \dots & \alpha_k^{-n/2} \\
									\alpha_1^{-n/2 +1} & \alpha_2^{-n/2 + 1} & \dots & \alpha_k^{-n/2 +1}\\
									\vdots & \vdots & \dots & \vdots \\
									\alpha_1^{n/2} & \alpha_2^{n/2} & \dots & \alpha_k^{n/2} 
									\end{array} \right ]$$
									
\noindent Throughout this section we will use slightly different notation where we will be interested in delta functions spaced at intervals of length $1/m$ and we will use $n$ to denote the number of measurements. Let $n = (1-\epsilon)m$. 

\begin{theorem}\label{thm:need}
Let $k = \Omega(\log m)$. Then the matrix $V_{-n/2, n/2}^k$ where $\alpha_j = e^{i 2 \pi j/m}$ has condition number at least $2^{\Omega(\epsilon k)}$.
\end{theorem}

Throughout this section let $V$ be short-hand for $V_{-n/2, n/2}^k$. It is easy to upper bound the maximum singular value of $V$ and hence the culprit is the smallest singular value. We conclude that there is a vector $u$ with $\|u\|_2 =1$ and $\| V u\|_2 \leq 2^{-\Omega(\epsilon k)}$. Then:
$$ \sum_{k =-n/2}^{n/2} \Big ( \sum_{j \mbox{ {\small is odd}}} u_j e^{i 2 \pi  j k/m} + \sum_{j \mbox{ {\small is even}}} u_j e^{i 2 \pi  j k/m} \Big )^2 = \|V u\|_2^2 \leq 2^{-\Omega(\epsilon k)}$$
And hence we have two sets of point sources corresponding to the odd and even indices respectively, with separation $\Delta = 2/m$ but where measuring them up to the cutoff frequency $n/2 = (1-\epsilon) m/2 = (1-\epsilon)/\Delta$ would require exponentially small error to tell them apart. 

\begin{corollary}\label{corr:impossible}
There is a pair of $k$ point sources $x$ and $x'$ each with separation $\Delta$ where if $m = (1-\epsilon)/\Delta$ and the noise $\eta_\ell$ has its real and complex part chosen as Gaussian independent random variables with variance $2^{-\Omega(\epsilon k)}$, then for any integer $\ell$ with $|\ell| \leq m$, $d_{TV}(v_{\ell}, v'_{\ell}) \leq 2^{-\Omega(\epsilon k)}$ where $v_{\ell}$ and $v'_{\ell}$ are the random variables associated with measuring $x$ and $x'$ at frequency $\ell$ respectively. 
\end{corollary}

\noindent Recall that in the noise-free case, the separation is irrelevant. It is possible to reconstruct the point sources with $m = k $. But in the noisy version of the problem, $k$ is irrelevant and the minimum number of measurements essentially only depends on the minimum separation. Let us now prove the main theorem in this section. We will make use of the Fejer kernel here (and later), although many other kernels could be used in its place.

\subsection{The Fejer Kernel I: Basic Properties}\label{sec:fejer}

Here we introduce the Fejer kernel, which we will use in the next subsection to construct a unit vector $u$ where $\|Vu\|_2$ is exponentially small. 

\begin{definition}
The Fejer kernel\footnotemark[4] is $$K_\ell(x) = \frac{1}{\ell^2} \sum_{j = -\ell}^\ell (\ell-|j|) e^{i 2 \pi j x} = \frac{1}{\ell^2} \Big ( \frac{\sin \ell \pi x}{\sin \pi x} \Big )^2$$
\end{definition}

\noindent Here we have given both the Fourier representation and the closed form for the Fejer kernel. We will use $\widehat{K}_\ell$ to denote its Fourier transform. The key is that $\widehat{K}_\ell$ is supported on $\{-\ell, -\ell + 1, ... \ell-1, \ell\}$ and yet $K_\ell(x)$ decays quickly. We will also make use of powers of the Fejer kernel which we denote by $K_\ell^r(x)$ and note that its Fourier transform is supported on $\{-r\ell, -r\ell + 1, ..., r\ell-1, r\ell\}$.

We make use of the following well-known property of the Fejer kernel:

\begin{fact}\label{fact:decay}
$K_\ell(x) \leq \frac{1}{4 \ell^2 x^2}$ for $x \in [-1/2, 1/2]$ and $K_\ell^r(x)\leq \frac{1}{4^r \ell^{2r} x^{2r}}$
\end{fact}

\begin{proof}
It is well known that $\sin \pi x \geq 2 x$ for $x \in [0, 1/2]$ and substituting this bound into the denominator and using $(\sin \ell \pi x)^2 \leq 1$ implies the bound on $K_\ell(x)$ and similarly for $K_\ell^r(x)$. 
\end{proof}

\footnotetext[4]{We use a slightly non-standard definition for the Fejer kernel, where ours will be normalized so that $K_\ell(0) = 1$ and periodic on $[0,1)$ instead of $[0, 2\pi)$ since this suits our applications more conveniently.}

\subsection{Lower Bounding the Condition Number}

We are now ready to prove Theorem~\ref{thm:need}. 

\begin{proof}
Consider $H(x) = K_\ell^r(x/m) e^{i \pi x}$. We will set $\ell = 4/\epsilon$ and $r = (k-1)/2\ell$.  We will establish some elementary properties of $H(x)$ and its Fourier transform $\widehat{H}$:

\begin{claim}
$\widehat{H}(t) = \sum_{j = k/2}^{3k/2 -1} h_j e^{i 2 \pi j t/m} $ and $\sum_j | h_j| = 1$. 
\end{claim}

\begin{proof}
Let $\widehat{K}_\ell$ denote the Fourier transform of $K_\ell(x/m)$. Then we can write $\widehat{H}$ as the convolution of $\widehat{K}_\ell$ with itself $r$ times, then convolved with $\delta(1/2)$. This directly implies the first part of the claim. Furthermore, the coefficients of $\widehat{K}_\ell$ are nonnegative and sum to one, so we can interpret both $\widehat{K}_\ell$  and $\widehat{H}$ as distributions on a discrete set of points, and this yields the second part of the claim. 
\end{proof}

We will let $u$ denote the vector whose entries are the $h_j$'s. Then $\|u\|_2 \geq \|u\|_1/\sqrt{k}= 1/\sqrt{k}$ because $u$ has $k$ non-zeros. Now we give an upper bound for $\|V u\|_\infty$. 

\begin{claim}
$\|V u\|_\infty \leq  \max_{|x| \leq n/2} |H(x)|$
\end{claim}

\begin{proof}
The claim is immediate since each coordinate in $Vu$ is equal to $H(j)$ for some $j \in \{-n/2, -n/2 + 1, ... n/2\}$. 
\end{proof}

\noindent Now all that remains is to bound the right hand side above. We can use Fact~\ref{fact:decay} and $K_\ell^r(x)\leq \frac{1}{ (8x/\epsilon )^{2r}}$. Moreover $H(x) = K_\ell^r(x/m - 1/2)$ and so $\max_{|x| \leq n/2} |H(x)| \leq \frac{1}{4^{2r}} \leq 2^{-\Omega(\epsilon k)}$ where we have used the fact that $K_\ell^r(x)$ has period one and that $r = \Omega(\epsilon k)$. This completes the proof of the theorem. 
\end{proof}

\section{Universal Preconditioning}

In the previous sections, we analyzed the condition number of $V_m^k$ based on $m$ and the minimum separation $\Delta$ according to the wrap-around distance, and this gave us algorithms for super-resolution that are {\em optimal} in the relationship between the number of measurements and the minimum separation. Here we pursue a different direction \-- namely, is there a way to precondition $V_m^k$ to make it almost orthogonal? Just like the Beurling-Selberg majorant played the key role in the previous section, here the Fejer kernel will itself be a good preconditioner for a family of Vandermonde matrices. 

\subsection{The Fejer Kernel II: Local Approximation}

In Section~\ref{sec:fejer} we introduced the Fejer kernel and stated some of its basic properties. Here we will also need need upper and lower bounds for $K_\ell(x)$ in a neighborhood of zero, that we establish here for later reference. Recall that
$$K_\ell(x)  = \frac{1}{\ell^2} \Big ( \frac{\sin \ell \pi x}{\sin \pi x} \Big )^2 \mbox{ and } K_\ell^r(x) = \Big (K_\ell(x) \Big )^r$$

\begin{lemma}\label{lemma:local}
For $\ell \geq 4$ and $|x| \leq 1/\ell$ and any integer $r \geq 1$ we have $$1 - C r \ell^2 x^2 \leq K_\ell^r(x) \leq 1 - c\ell^2 x^2$$
where we can take $C = 12$ and $c = 1/3$. 
\end{lemma}

\begin{proof}

In order to analyze the behavior of $K_\ell^r(x)$ in a neighborhood of zero, we need upper and lower bounds for $\sin x$ that behave like $x - \Theta(x^3)$. We will make crucial use of the following facts, which are easy to check:

\begin{fact}
$x(1 - \frac{x^2}{2}) \leq \sin x $ for all $x \geq 0$
\end{fact}

\begin{fact}\label{fact:upper}
$\sin x \leq x(1 - \frac{x^2}{10}) $ for all $x  \in [0, \pi)$
\end{fact}

\noindent Then substituting these inequalities we obtain 
$$K_\ell^r(x)  \leq \frac{1}{\ell^{2r}}  \Big ( \frac{ \ell \pi x (1 - (\ell \pi x)^2/10)}{\pi x (1 - (\pi x)^2/2)} \Big )^{2r} = 
\Big ( \frac{  1 - (\ell \pi x)^2/10}{ 1 - (\pi x)^2/2} \Big )^{2r} \leq \Big ( \frac{  (1 - (\ell \pi x)^2/10)}{ (1 - (\pi x)^2/2)} \Big )$$
where we have used the facts above, and we require $|x| \leq 1/\ell$ in order to invoke Fact~\ref{fact:upper}. Finally we have
$$K_\ell^r(x)  \leq 1  - (\ell \pi x)^2/10  + (\pi x)^2 \leq 1 - \ell^2 x^2/3$$
where we have used that $(1 - (\pi x)^2/2)^{-1} = 1 + (\pi x)^2/2 + ... \leq 1 + (\pi x)^2$ which holds for $x \leq 1/\pi$. Finally we can prove our lower bound 
$$K_\ell^r(x)  \geq \frac{1}{\ell^{2r}}  \Big ( \frac{ \ell \pi x (1 - (\ell \pi x)^2/2)}{\pi x (1 - (\pi x)^2/10)} \Big )^{2r} = 
\Big ( \frac{  1 - (\ell \pi x)^2/2}{ 1 - (\pi x)^2/10} \Big )^{2r}  \geq \Big  (1 - (\ell \pi x)^2/2)^{2r} \geq 1 - r (\ell \pi x \Big )^2$$
where again we have used the facts above, and we require that $|x| \leq 1/\pi$ in order to invoke Fact~\ref{fact:upper}. 
\end{proof}








\subsection{From Uncertainty Principles to Preconditioners}\label{sec:salem}

Let us contrast our approach here with the way in which we used majorants in the previous section: Recall our goal was to upper and lower bound the sum $\sum_{t = -m}^m \|f(t)\|^2$. Once we plugged in majorants and minorants, both the upper and lower bounds became simple expressions but the bounds themselves are not that close to each other unless $m$ is much larger than $1/\Delta$ and this is roughly tight. Instead, we could ask: are there coefficients $c_t$ so that the sum $\sum_{t = -m}^m c_t \|f(t)\|^2$ itself has a simple expression that we can get tighter upper and lower bounds for (see e.g. Claim~\ref{claim:rep}). 

The point is that many inequalities in harmonic analysis are established this way. 
One classic example is the Salem inequality (see e.g. \cite{Z}). This can be thought of as a strengthening of some of the more well-known uncertainty principles \cite{DS}, which applies to the setting where a signal is sparse \mbox{and} its spikes are well-separated. Again let $f(t) = \sum_{j = 1}^k u_j e^{i 2 \pi  f_j t}$ and suppose that $| f_j - f_{j'}| \geq \Delta$ for all $j \neq j'$. Then

\begin{lemma} \cite{Z}
Let $E$ be any interval of length $(1+\epsilon)/\Delta$. Then $\sum_j |u_j|^2 \leq \frac{C_\epsilon}{|E|} \int_E \|f(t)\|^2 \mbox{ } dt$ where $C_\epsilon = \frac{2 \pi (1 + \epsilon)^2}{4 \epsilon (2 + \epsilon)} = O(1 + 1/\epsilon)$.
\end{lemma}

\noindent In fact the proof of this inequality follows by choosing an appropriate test function $\chi(t)$ that is supported on $E$, and analyzing $$ \int_{-\infty}^\infty \chi(t) \|f(t)\|^2 \mbox{ } dt$$ directly and showing that it equals (up to scaling) $\sum_j |u_j|^2 $ plus cross terms that decay as the separation increases. So when we use a single test function (as opposed a majorant and a minorant) it can really be thought of as a preconditioner that makes the right hand side close to the sum-of-squares of the coefficients!

We will pursue this direction, but will make use of different test functions than in \cite{Z} that are supported on a discrete set of points, namely the Fourier transform of the Fejer kernel and its powers. This has the effect of transforming the right hand side into a sum instead of an integral.

\begin{claim}\label{claim:rep}
Let $F(x)$ be a kernel where $\widehat{F}(t) = \sum_{j = -m}^m c_j \delta(t - j)$. Then $$\sum_{j = -m}^m c_j \| f(j)\|^2 = \sum_{j = 1}^k \sum_{j' = 1}^k u_j \bar{u}_{j'} F(f_j - f_{j'})$$
\end{claim}

\begin{proof}
$$\sum_{j = -m}^m c_j \| f(j)\|^2 = \sum_{j = 1}^k \sum_{j' = 1}^k u_j \bar{u}_{j'} \int_{-\infty}^{\infty} \widehat{F}(t) e^{i 2 \pi  (f_j - f_{j'}) t)} \mbox{ } dt 
= \sum_{j = 1}^k \sum_{j' = 1}^k u_j \bar{u}_{j'} F(f_j - f_{j'})$$
\end{proof}

\noindent Now we can set $F$ to be the Fejer kernel or its powers. Using the above claim we will show that this choice of $c_j$'s is a good preconditioner. Suppose that $d_w(f_j, f_{j'}) \geq \Delta$. 

\begin{lemma}\label{lemma:salem}
If $c_j$ are the coefficients of $\widehat{K}_\ell$, then $$\Big | \sum_{j = -\ell}^\ell c_j \| f(j)\|^2 - \sum_{j = 1}^k |u_j|^2 \Big | \leq \sum_{j = 1}^k  |u_j|^2 \frac{ \pi^2}{24 \ell^2 \Delta^2}$$
\end{lemma}

\begin{proof}
We can use Claim~\ref{claim:rep} to obtain $$ \Big | \sum_{j = -\ell}^\ell c_j \| f(j)\|^2  - \sum_{j = 1}^k |u_j|^2 \Big | \leq \sum_{j \neq j'}  |u_j| |u_{j'}| | K_\ell(f_j - f_{j'})|$$ where we have used the fact that $K_\ell(0) = 1$ (recall that we chose a non-standard normalization for the Fejer kernel). Moreover rearranging we get: 
\begin{eqnarray*}
\sum_{j \neq j'}  |u_j| |u_{j'}| | K_\ell(f_j - f_{j'})|& \leq& \sum_{j \neq j'} \frac{1}{2}( |u_j|^2 + |u_{j'}|^2) | K_\ell(f_j - f_{j'})| \\
 &\leq& \sum_{j = 1}^k  |u_j|^2 \sum_{j' = 1}^\infty \frac{1}{4 \ell^2 (j')^2 \Delta^2} = \sum_{j = 1}^k  |u_j|^2 \frac{ \pi^2}{24 \ell^2 \Delta^2}
\end{eqnarray*}
where have used Fact~\ref{fact:decay} in the last inequality. This implies the lemma. 
\end{proof}

\begin{corollary}\label{cor:salem}
If $c'_j$ are the coefficients of $\widehat{K}^r_\ell$ then $$\Big | \sum_{j = -r\ell }^{r\ell} c'_j \| f(j)\|^2 - \sum_{j = 1}^k |u_j|^2 \Big | \leq \sum_{j = 1}^k  |u_j|^2 \frac{ \pi^2}{ 4^r \ell^{2r} \Delta^{2r}}$$
\end{corollary}

\begin{proof}
We can follow the same proof as above, but for powers of the Fejer kernel. The only difference is that we use fact that $K_\ell^r(\Delta) \leq \frac{1}{4^r \ell^{2r} \Delta^{2r}}$ and that $\sum_{j =1}^{\infty} \frac{1}{j^{2r}} \leq \sum_{j =1}^\infty \frac{1}{j^2} = \frac{\pi^2}{6}$. This implies the corollary.
\end{proof}


\subsection{Approximate Strong Convexity}

In the previous subsection, we proved upper and lower bounds for $\sum_{j = -r\ell}^{r\ell} c'_j \| f(j)\|^2$, and eventually we want to use this potential function to guide our algorithm. Let us focus on the case where all the coefficients $a_j = 1$ for simplicity.  What we want to prove are soundness and completeness results that $$\sum_{j = -r\ell}^{r\ell} c'_j \| f(j) - e^{i 2  \pi  j x}\|^2 \approx k - 1$$ if and only if $x$ is sufficiently close to some $f_j$. So the challenge is that we cannot directly apply Lemma~\ref{lemma:salem} or Corollary~\ref{cor:salem} because the set $\{x\} \cup \{f_j\}_j$ will not necessarily be well separated (and we hope it won't be). 

The starting point is combining Claim~\ref{claim:rep} and Corollary~\ref{cor:salem} to obtain:
$$\sum_{j = -r\ell}^{r\ell} c'_j \| f(j) - e^{i 2 \pi j x}\|^2 = T  -  2 \sum_{j=1}^k  K_\ell^r(x - f_j)$$
where $T = \sum K_\ell^r(f_j - f_{j'}) + K_\ell^r(0) \approx k + 1$ is a constant that does not depend on $x$, and we have used the fact that $K_\ell^r$ is symmetric about the origin. Let $\gamma = d_w(x, \{f_j\})$. 

\begin{lemma}\label{lemma:structure}
If $\gamma < \min(\Delta/2, 1/\ell)$ and $\ell \geq 4$ then
$$1 - 12 r \ell^2 \gamma^2 - \frac{\pi^2}{4^r \ell^{2r} \Delta^{2r}} \leq \sum_{j=1}^k  K_\ell^r(x - f_j) \leq 1 - \frac{  \ell^2 \gamma^2}{3} + \frac{\pi^2}{4^r \ell^{2r} \Delta^{2r}}$$
And if instead $\gamma \geq 1/\ell$ then
$$ \sum_{j=1}^k  K_\ell^r(x - f_j) \leq \frac{1}{4^r} + \frac{\pi^2}{4^r \ell^{2r} \Delta^{2r}}$$
\end{lemma}

\begin{proof}
Let $f_1, f_2, ... ,f_k$ be sorted so that $d_w(x, f_j)$ is monotonically decreasing. We can write 
$$ \sum_{j=1}^k K_\ell^r(x - f_j) = K_\ell^r(\gamma) + \sum_{j = 2}^k K_\ell^r(x - f_j)  \leq 1 - \frac{ r \ell^2 \gamma^2}{5}+ \sum_{j = 2}^k K_\ell^r(x - f_j)$$
where we have used  Lemma~\ref{lemma:local} to give upper bounds for $K_\ell^r(\gamma)$. Moreover we have that and $d_w(x, f_j) \geq \Delta (j-1) - \gamma \geq \Delta (j-1)/2$ for all $j$ and so we can use Fact~\ref{fact:decay} as we did in the proof of the modified Salem inequality (Lemma~\ref{lemma:salem}) and this implies the upper bound. An identical argument implies the lower bound too. The second part of the lemma follows by instead invoking Fact~\ref{fact:decay}. 
\end{proof}

\noindent Notice that if we set $r = \Theta(\log n)$ and $\ell > \frac{1}{\Delta}$ then the terms that depend on $\Delta$ can be bounded by any inverse polynomial and we have proven that our potential function is additively close to being strongly convex. The intuition is that if we have any two solutions $x$ and $x'$ at distance $\gamma$ and $\gamma'$ respectively from $f_1$ if $\gamma < \frac{\gamma'}{C \ell}$ our potential function must be smaller for $x$ than it is for $x'$, and so we iteratively discard points that are too far from $f_1$ and make geometric progress in each step. 

\subsection{The Algorithm}

Here we consider the following abstraction: Suppose there are unknown $f_j$'s which are $\Delta$ separated according to the wrap-around metric and set $\ell = 1/\Delta$. Furthermore we are given an oracle $F$ so that if $d_w(z, \{f_j\}) = \gamma$, then either 
\begin{equation}\label{eqn:oracle}
A + c\ell^2 \gamma^2 - \epsilon^2/4 \leq F(z) \leq A+ C r \ell^2 \gamma^2 + \epsilon^2/4\tag{$\ast$}
\end{equation}
if $\gamma \leq 1/\ell$ and otherwise $F(z) \leq 1/4$. Our goal is to find an approximation $\widehat{f}_j$ to the $f_j$'s using this oracle, that have matching distance at most $\epsilon$. This will immediately yield a fast, greedy algorithm for the super-resolution problem. 

\begin{fragment*}[t]
\caption{
\label{alg:iterrefine}{\sc IterativeRefinement},  \textbf{Input:} Oracle $F$ and $k$ and constants $C, c, \epsilon$ \\   \textbf{Output:} $\{\widehat{f}_j\}$ with matching distance at most $\epsilon$ to $\{f_j\}$  \vspace*{0.01in}
}

\begin{enumerate} \itemsep 0pt
\small 
\item Set $\delta_j =  \frac{1}{4\ell } ( \frac{c}{C r}  )^j$ and $\gamma_j =  (\frac{C r}{c} )^{1/2} \delta_j + \epsilon/2$
\item Set $U$ to be the grid of $[0, 1)$ with length $\delta_1$
\item For $i = 1$ to $k$
\item  $\quad$ Find $z$ in $U$ that minimizes $F(z)$;  Set $\widehat{f}_i$ to $z$
\item  $\quad$ Remove all $x'$ with $d_w(\widehat{f}_i, x') \leq 1/(2\ell)$
\item Set $j = 2$
\item While $\gamma_j > \epsilon$
\item  $\quad$ For $i = 1$ to $k$
\item  $\quad$ $\quad$ Let $U_i$ be the grid of points $x'$ with $d_w(\widehat{f}_i, x') \leq \gamma_{j-1}$ with length $\delta_j$. 
\item  $\quad$ $\quad$ Find $z$ in $U_i$ that minimizes $F(z)$;  Set $\widehat{f}_i$ to $z$
\item Output $\{\widehat{f}_i\}$
\end{enumerate} 

\end{fragment*}

Let $C, c$ be constants, then:

\begin{theorem}\label{thm:strong}
The output of {\sc IterativeRefine} is a set $\{\widehat{f}_i\}$ that has matching distance at most $\epsilon$ to $\{f_i\}$ with respect to $d_w$. If $R$ is the running time to call the oracle for $F(z)$, then {\sc IterativeRefine} runs in time $O(R \ell r + R k r^{1/2} \log (1/\epsilon))$. 
\end{theorem}

\begin{proof}
We will prove that {\sc IterativeRefinement} succeeds by proving a set of invariants for the algorithm. We will refer to Steps $2$-$6$ as the initialization phase. 

\begin{claim}
After initialization, the matching between $\{\widehat{f}_i\}$ and $\{f_i\}$ according to $d_w$ is at most $\gamma_1 \leq 1/(4\ell)$. 
\end{claim}

\begin{proof}
Consider the first iteration where $i = 1$. Since the grid length is $\delta_1$, there must be an $x$ where $ A + Cr \ell^2 \delta_1^2 - \epsilon^2/4 \leq F(x)$. Hence $z$ must satisfy $$ A + Cr\ell^2 \delta_1^2 - \epsilon^2/4 \leq F(z) \leq A + c \ell^2 \gamma^2 + \epsilon^2/4$$
where $d_w(\{f_i\}, z) = \gamma$. This implies that $\gamma \leq \gamma_1 =  (\frac{Cr}{c} )^{1/2} \delta_1 + \epsilon/2 \leq \frac{1}{4\ell}$. Hence if $f_1$ is the closest to $z$ according to $d_w$, we necessarily remove all points $x'$ with $d_w(x', f_1) \leq 1/(4\ell)$. At the start of the next iteration, there still remains a point $z$ that is within $\delta_1$ of some $f_j$ (even after we delete points close to $\widehat{f}_1$) and so the next point $z$ will be close to a new $f_j$. This establishes the claim by induction. 
\end{proof}

\begin{claim}
After iteration $j$, the matching distance between  $\{\widehat{f}_i\}$ and $\{f_i\}$ is at most $\gamma_j$. 
\end{claim}

\begin{proof}
By induction, at the start of iteration $j$ the matching distance between  $\{\widehat{f}_j\}$ and $\{f_j\}$ is at most $\gamma_{j-1}$. Thus in iteration $j$, for each $i$ there is a point that is within distance $\delta_j$ of $f_i$. Moreover the grids $U_i$ are disjoint. Again using the properties of the oracle, this implies that for each $i$ the $z$ we find satisfies $d_w(f_i, z) = \gamma$ with $$\gamma \leq \Big (\frac{C r}{c} \Big )^{1/2} \delta_{j-1} + \epsilon/2 = \gamma_j$$ and this completes the proof of the claim. 
\end{proof}

These invariants immediately yield our main theorem. 
\end{proof}

We are now ready to prove Theorem~\ref{thm:last}.

\begin{proof}
We set $\ell = 1/\Delta$ and $r = \Theta(\log 1/\epsilon)$. Recall that $m = r\ell$ is the cut-off frequency. Furthermore let
$$F(z) = \sum_{j = -r\ell}^{r\ell} c'_j \| f(j) - e^{i 2 \pi j z}\|^2 $$ in which case $A = T - 2$ and using Lemma~\ref{lemma:structure} we conclude that $F(z)$ has the desired form in \eqref{eqn:oracle}. Moreover we can evaluate $F(z)$ in time $O(m)$ by directly computing it. Finally it is easy to see that if there is noise in our measurements we can account for it expanding $F(z)$ as follows:
$$\Big | \sum_{j = -r\ell}^{r\ell} c'_j \| f(j) - e^{i 2 \pi j z} + \eta_j \|^2  - F(z) \Big | \leq \sum_{j = -r\ell}^{r\ell} c'_j \Big( 2 (k+1) \|\eta_j\| + \|\eta_j\|^2 \Big)$$
Note that $\sum_j c'_j = 1$ because it is the convolution of discrete probability distributions. Hence if the noise satisfies $\|\eta_j\| \leq \epsilon^2/(4k)$ for each $j$, we can group the error terms due to the noise into the terms in \eqref{eqn:oracle}  that depend on $\epsilon$, and the theorem now follows by invoking Theorem~\ref{thm:strong}.
\end{proof}


\subsection*{Acknowledgements}

Thanks to David Jerison and Henry Cohn for many helpful discussions about extremal functions.

\end{document}